\documentclass[submission,copyright,creativecommons]{eptcs}
\providecommand{\event}{AiML 2026}
\usepackage{aiml26}

\usepackage{iftex}

\ifpdf
  \usepackage{underscore}         
  \usepackage[T1]{fontenc}        
\else
  \usepackage{breakurl}           
\fi

\title{On Modal Logics of Connectedness in Metric Spaces}
\author{John Harding \qquad\qquad  Ilya Shapirovsky 
\institute{New Mexico State University\\ NM, USA}
}

\newcommand{\titlerunning}{On Modal Logics of Connectedness in Metric Spaces}
\newcommand{\authorrunning}{John Harding, Ilya Shapirovsky}

\hypersetup{
  bookmarksnumbered,
  pdftitle    = {\titlerunning},
  pdfauthor   = {\authorrunning},
  pdfkeywords = {modal logic, distance logic, 
 connected metric space,  connected graph, finite model property} 
}

\usepackage{scalerel}
\usepackage{tikz}
\usepackage{amsfonts,amscd,amsmath,amsthm,amssymb,amsfonts,latexsym,}
\usepackage{xcolor}
\usepackage{enumerate}
\usepackage{manfnt}

\def\AA{\forall}
\def\imp{\to}
\def\toto{\twoheadrightarrow}

\def\vf{\varphi}
\def\mo{\vDash}

\def\emp{\varnothing}

\def\clK{\mathcal{K}}

\def\clK{\mathcal{K}}

\def\Log{\logicts{Log}}

\newcommand\logicts[1]{\mathrm{#1}}

\def\mR{\mathbb{R}}

\def\mQ{\mathbb{Q}}

\newcommand\LS[1]{\logicname{S#1}}

\theoremstyle{definition}
\newtheorem{problem}{Problem}
\newtheorem*{problem*}{Problem}
\newtheorem*{remark*}{Remark}

\def\EE{\exists}
\def\AA{\forall}

\def\Log{\mathop{\mathrm{Log}}}

\newcommand\logicname[1]{\logicts{#1}}
\newcommand\LMetrAx[1]{\logicname{Metr}{(#1)}}

\def\Fms{\logicname{Fms}}

\def\Di{\lozenge}

\newcommand\alphabetname[1]{{\mathrm #1}}
\def\Alp{\alphabetname{A}}
\def\AlA{\Alp}
\def\Al{\Alp}

\def\AlpB{\alphabetname{B}}
\def\AlB{\AlpB}

\def\Par{\alphabetname{A}}

\def\d{d}

\def\dj{\d_{\j}}

\def\ups{{\uparrow}}
\def\ds{{\downarrow}}

\def\eps{\varepsilon}
\newcommand\idx{\mathrm{ind}}

\newcommand\myoper[1]{\mathop{\mathrm{#1}}}
\def\Sub{\myoper{Sub}}

\newcommand\sumvct[1]{\oplus{\vect{#1}}}
\newcommand\vect[1]{{\vec{#1}}}
\def\vct{{\vect{r}}}
\def\sumv{\sumvct{r}}

\def\val{v}

\newcommand{\ff}[1]{\widehat{#1}}
\def\ACon{\logicts{Con}}

\def\cl{\mathrm{c}}
\usepackage{graphicx}

\usetikzlibrary{fit}

\newcommand\ISLater[1]{{\color{gray}~IS: #1}}
\renewcommand\ISLater[1]\empty
\newcommand\hide[1]\empty

\begin{document}
\maketitle
\begin{abstract}


For a positive number $a$, each metric space carries the relation $D_a$ consisting of those pairs that are of distance less than $a$ apart. A space $X$ is said to be {\em $a$-connected}, if the graph $(X,D_a)$ is connected (that is, there is a $D_a$-path between every pair of points in $X$). 
We give a complete axiomatization of $a$-connected metric spaces in the language with a family of distance modalities and  the universal modality. Then we give a complete axiomatization of the logic of connected (in the classical topological sense) metric spaces in the language with the topological modality, universal modality, and a single distance modality.  We also show that these logics have the finite model property.


\end{abstract}

\section{Introduction}
Let $(X,d)$ be a metric space, which we usually refer to as simply $X$. There is a natural way to associate a relational structure, or Kripke frame, with $X$. For each real number $r>0$ define a relation $D_r$ on $X$ by setting
\[
x\,D_r\, y\,\,\Leftrightarrow\,\, d(x,y)<r.
\] 
This allows us to interpret in $X$ modal formulas with the corresponding distance modalities 
$\Di_r$. 
We can also consider the topological closure operation of $X$, 
which allows us to 
interpret in $X$  modal formulas with the closure modality $\Di$.

The situation described is the general setting of \cite{Kutz-Sturm-Suzuki-Wolter-Zakharyaschev2002, Kutz-Sturm-Suzuki-Wolter-Zakharyaschev2003,Wolter2005,KuruczWZ05,Kutz2007}. In this series of papers, numerous results were established related to the axiomatizability and decidability of modal logics of metric spaces with different families of operators. For instance, the axiomatization of the logic of all metric spaces with distance modalities 
$\Di_r$ for each $r>0$ is given by the family of axioms for all $r,s>0$ \cite{Wolter2005}:
\vspace{1ex}

\begin{itemize}
\item[(1)] $p\,\to\,\Di_r p$
\item[(2)] $p\,\to\,\Box_r\Di_rp$
\item[(3)] $\Di_r\Di_s p\to\,\Di_{r+s} p$
\end{itemize}
\vspace{1ex}

\noindent The first axiom says that for a subset $P$ of a metric space, we have $P$ is contained in the set of points of distance less than $r$ from $P$; the second expresses the symmetry of the metric; and the third comes from the triangle inequality. 

In this paper, we are interested in modal logics of connected  metric spaces. More precisely, we consider two forms of connectedness: in the usual topological sense,  and in a weaker sense of connectedness of the graph $(X,D_a)$ for some fixed positive $a$.  We address the latter property 
as {\em $a$-connectedness}, which means that there is a $D_a$-path between every pair of points in $X$. 

It is known that connectedness of a topological space is expressed in the language with the closure and universal modalities by the formula 
\begin{equation}\label{eq:intro-conn}
\EE p \wedge \EE \neg  p \imp \EE (\Di  p\wedge \Di  \neg p)    
\end{equation}
 and  moreover, a complete 
axiomatization of the class of all connected spaces is the logic $\LS{4UC}$,  the extension of $\LS{4}$
with the axioms of universal modality and the above formula \cite{Shehtman99}.  
The axiomatization of connected metric spaces  
in the language of the topological closure $\Di$, universal modality, and distance modalities
is an open problem \cite{Wolter2005,KuruczWZ05}. 
We make partial progress in this direction, providing the following two completeness results. 
\begin{itemize}
    \item For a given positive $a$, we identify an axiomatization  of $a$-connectedness 
in the language without topological closure, with the universal modality, and with any set of distance  modalities 
$\Di_r$ containing $\Di_a$. 
\item 
We give an axiomatization of    topological connectedness in the language of closure $\Di$, universal modality, and a single
distance modality $\Di_r$.
\end{itemize}
We also show that 
these logics have the finite model property and, in the case of a finite language, are decidable.  

The paper is organized as follows.   Syntactic and semantic conventions are given in Section \ref{sec:syn-mod-conn}. In Section \ref{sec:fmp}, we prove the finite model property theorem, one of the ingredients 
of the completeness proofs. In Section \ref{sec:a-conn}, we give the axiomatization of $a$-connected spaces, and the result for topological connectedness is given in Section \ref{sec:top-conn}.

\section{Logics and models}\label{sec:syn-mod-conn} 


\subsection{Language}

Let $\Par$ be a set of positive real numbers ({\em parameters}), which will also be considered as indices of modalities.
An  {\em $\Par$-logic} is a normal modal logic with modalities $\{\Di_a\mid a\in\Par\}\cup\{\EE\}$,
and an  {\em $(\Par,\Di)$-logic} is a logic in this  language endowed with $\Di$,
where
$\EE$ and $\Di$ are two fresh symbols.
We write $\AA$ for the abbreviation $\neg\, \EE\, \neg$.
In relational or topological spaces,  $\EE$ will be interpreted as ``somewhere'', that is, formally, via the universal relation.  
In topological spaces, $\Di$ will be interpreted as the operation of closure. 


For a logic $L$ and a formula $\vf$, 
the smallest logic containing $L\cup\{\vf\}$ is denoted as 
$L+\vf$.

\subsection{Axioms}

Let $\vect{r}=(r_1,\ldots,r_m)$ be a tuple of parameters. 
We write $\Di_{\vct}$ for the compound modality $\Di_{r_1}\ldots \Di_{r_m}$, and  
 $\sumvct{r}$ for  $r_1+\ldots+r_m$.
We put $\sumvct{r}=0$ for the empty $\vect{r}$.

\begin{definition}\label{def:metr-axioms}
We define $\LMetrAx{\Par}$ as the smallest normal $\Par$-logic that, for each $r\in \Par$, contains the formulas
\begin{enumerate}[(i)]
    \item  $p\imp \AA \EE p$, $p\imp\EE p$, $\EE \EE p \imp \EE p$, and   $\Di_r p\imp \EE p$;
     \item $p\imp \Box_r \Di_r p$;
     \item $\Di_{\vct} p\imp \Di_{r} p$ for each tuple $\vct$ over $\Par$  with $\sumv\leq r$.
 \end{enumerate}

\noindent Then set  $\LMetrAx{\Par,\Di }$ to be the smallest normal $(\Par,\Di)$-logic that contains $\LMetrAx{\Par}$ and the formulas 
\begin{enumerate}[(i)]
    \item  $p\imp \Di p$, $\Di \Di p\imp \Di  p$, $\Di p\imp \EE p$, and
    \item $\Di_r\Di p \imp \Di_r p$ for each $r\in \Par$.
 \end{enumerate}
\end{definition}
\begin{remark}
    Note that $\Di p\imp \EE p$ follows from other axioms, provided that $\Al$ is non-empty. 
    For the empty $\Al$, this logic is known as $\LS{4U}$, the logic of all (finite) preorders endowed with the universal relation \cite[Theorem 5.9]{Goranko-Passy}. 
\end{remark}

\subsection{Semantics}

In the following, we recall that our language includes a fixed set $A$ of positive real numbers as parameters.  

\begin{definition}\label{def:MetricFrame}
Let $S=(X,d)$ be a metric space. For $r\in A$, 
let $D_{r}$ be the binary relation on $X$ given by  $(x,y)\in D_{r}$ iff
$d(x,y) < r$.
The {\em metric $\Par$-frame} 
of $(X,d)$ 
is the relational structure
$(X,(D_{r})_{r\in \Alp},X\times X)$.
%
Let $\Log{F_\Alp(S)}$ denote the set of all $\Al$-formulas that are valid in the metric $A$-frame of $S$. The {\em $\Par$-logic} of a class $\clK$ of metric spaces is defined as $\bigcap\{\Log{F_\Alp(S)}\mid S\in \clK\}$.
%
To define the {\em $(\Par,\Di)$-logic of $\clK$}, we additionally interpret $\Di$ as the topological closure.
\end{definition}

In \cite{Wolter2005}, it was shown that $\LMetrAx{\Par,\Di }$ is the $(\Par,\Di)$-logic of the class of all metric spaces.

\begin{remark}
In \cite{Wolter2005}, the axiom system was formally different, since extra conditions were assumed for the set of parameters. The difference is minor.
\end{remark}

\subsection{Expressing connectedness}  
A unimodal frame $F=(X,R)$ is {\em connected}, if for any points $x,y$ in $X$, there is a non-oriented $R$-path from $x$ to $y$, that is: there are $n<\omega$ and points $x_0=x, x_1, \ldots, x_n=y$ such that for each $k<n$, $x_k \,R\, x_{k+1}$ or $x_{k+1}\, R\, x_{k}$.
A frame $F=(X,(R_i)_{I})$ is {\em $a$-connected} for $a\in I$,  if $(X,R_a)$ is connected.

%

The  property of  $a$-connectedness is  expressed by the formula 
\begin{equation}\label{eq:conn-general} 
\EE p \wedge \EE \neg  p \imp \EE ((p\wedge \Di_a \neg p)\vee  (\neg p\wedge \Di_a p))     
\end{equation}
Here $\EE$ is interpreted as the modality of the universal relation $X\times X$. 
That is, we have:
$$
(X,(R_i)_{I}) \text{ is  $a$-connected} \text{ iff } 
\text{$\ACon_a$ is valid in $(X,(R_i)_{I},X\times X)$.}
$$ 
The proof is straightforward.  

For a symmetric $R_a$, \eqref{eq:conn-general} simplifies to
$$
\EE p \wedge \EE \neg  p \imp \EE (p\wedge \Di_a \neg p).
$$
This formula will be denoted by $\ACon_a$.

For the case when $R_a$ is reflexive,
\eqref{eq:conn-general} 
simplifies to
\begin{equation}\label{eq:conn-top} 
\EE p \wedge \EE \neg  p \imp \EE (\Di_a p\wedge \Di_a \neg p).    
\end{equation}

\smallskip
If the modality is interpreted as the closure in a topological space, 
\eqref{eq:conn-top}  expresses  topological connectedness  \cite{Shehtman99}; 
in this case, we denote this formula as $\ACon_\Di$.\footnote{More ways to represent topological connectedness in propositional languages are discussed in \cite{RCC-2008,RCC-2010}. }

For the empty $\Al$, the logic $\LMetrAx{\Par,\Di }+\ACon_\Di$  is known as $\LS{4UC}$. 
By \cite[Theorem 10]{Shehtman99}, this logic has the finite model property: it is characterized 
by the class of finite connected preorders endowed with the universal relation. 
Moreover, $\LS{4UC}$  is the logic of all connected topological spaces, and in fact is 
characterized by any connected 
dense-in-itself
separable metric space  
(i.e., any connected metric space with more than one point)  \cite[Theorem 18]{Shehtman99}.

\section{The finite model property of logics of connectedness}\label{sec:fmp}

In this section, we consider logics in an abstract setting, and their models are  not assumed
to be related to any geometric spaces -- they are just relational structures.

\subsection{Filtration lemma}

In this subsection, we consider  models, formulas, and logics in a fixed  modal alphabet. 



For a set of formulas $\Gamma$, let $\Sub \Gamma$ be the set of subformulas of formulas occurring in $\Gamma$. If $\Gamma=\Sub \Gamma$, then $\Gamma$ is said to be {\em $\Sub$-closed}. For a model $M$ and a set $\Gamma$  of formulas, let  $\sim_\Gamma$ be the equivalence on $M$ induced by $\Gamma$:
\begin{center}
$x\sim_{\Gamma} y$ \quad iff \quad$\forall \psi\in\Gamma\; (M,x\models \psi \text{ iff } M,y\models \psi)$.
\end{center} 
We next recall the definition of the minimal and maximal filtrations of a relation.


\newcommand\maxR[1]{R^{\Gamma}_{{#1},\sim}}
\newcommand\maxRDel[1]{R^{\Delta}_{{#1},\sim}}

\begin{definition}\label{def:epi}

Suppose $\Gamma$ is a $\Sub$-closed set of formulas,
$M=(X,(R_a)_\Alp,\val)$ is a model, and $\sim$ is an equivalence relation that refines $\sim_\Gamma$. Set $\ff{X}=X/{\sim}$. 
For a relation $R$ corresponding to a modality $\Di$, we 
define the {\em minimal filtered relation}  $R_{\sim}$ and the \emph{maximal filtered relation} 
$R^{\Gamma}_\sim$
on $\hat{X}$ by setting 
\begin{align*}
\,\,\,\,[x]\,{R}_{\sim}\,[y]\,\,  & \quad \text{ iff }\quad  \mbox{there exist $x'\sim x$ and $y'\sim y$ with } x'\,R\,y'\\
[x]\,R^{\Gamma}_\sim\,[y]\,\,  &\quad \text{ iff } \quad  \mbox{for all $\psi$ with $\Di  \psi\in \Gamma$, if $M,y\models \psi$ then $M,x\models \Di \psi$}.
\end{align*}
\end{definition}


It is easy to check that $R_{\sim}$ is contained in  $R^{\Gamma}_\sim$. 

\begin{definition}
\label{defn: filtration}
A {\em filtration of $M$ through $\Gamma$} is a model $\ff{M}=(\ff{X},(\ff{R}_a)_\Alp,\ff{\val})$ such that
\begin{enumerate}[(i)]
\item $\ff{X}=X/{\sim}$ for an equivalence relation $\sim$, which refines $\sim_\Gamma$;
\item  $\ff{\val}(p)=\{[x]\in \ff{X} \mid x \in\val(p)\}$ for  all variables $p\in \Gamma$;
\item For each $a\in \Al$,  $R_{a,\sim} \subseteq \ff{R}_a  \subseteq  \maxR{a}$.
\hide{
, where
$$
\begin{array}{ccl}
~[x]\,{R}_{a,\sim}\,[y] & \text{iff} & \exists x'\sim x\ \exists y'\sim y\;
(x'\,R_a\,y'),
\smallskip \\
~[x]\,\maxR{a}\,[y] & \text{iff} & \forall \psi\;
   (\Di_a \psi\in \Gamma \: \& \: M,y\models \psi \Rightarrow M,x\models \Di_a \psi ).
\end{array}
$$}
\end{enumerate}
\end{definition}



The following fact is standard.
\begin{lemma}[Filtration lemma]\label{Lemma:Filtration} Suppose that $\Gamma$ is a  $\Sub$-closed set of formulas and $\ff{M}$ is a $\Gamma$-filtration of a model~$M$. Then, for all points $x$ in $M$ and all formulas ${\psi \in\Gamma}$, we have:
\begin{center}
${M,x\models \psi}\quad$ iff $\quad{\ff{M},[x]\models\psi}$.
\end{center}
\end{lemma}
\begin{proof}
Induction on $\psi$.
\end{proof}

For a relation $R$ on $X$,
let $R^+$ be its transitive closure $\bigcup_{0<i<\omega}R^i$.
The following is well known. 

\begin{lemma}(\cite[Lemma 2]{Seg1968-S4.1})\label{lemma:poly:transFiltr}
 Let $M\mo \Di_a \Di_a \psi \imp \Di_a \psi $ for each formula $\psi$,  and let
$\sim$ be the equivalence induced by a $\Sub$-closed $\Gamma$ on $M$.
Then $(R_{a,\sim})^+$ is contained in the maximal filtered relation $\maxR{a}$.
\end{lemma} 
\begin{proof}
By induction on $n$,  we show $(R_{a,\sim})^n\subseteq \maxR{a}$ for all $n>0$. Basis is trivial, since the $a$-th minimal filtered relation is contained in the $a$-th maximal.
For the step, suppose
 $\Di_a \psi \in \Gamma$,  $M,y\mo  \psi$, and
$[x] (R_{a,\sim})^{(n+1)} [y]$. 
Then $[x] R_{a,\sim} [z] (R_{a,\sim})^n [y]$ for some $z$. By the hypothesis,
$M,z\mo \Di_a \psi$.  We have $x'R_a z'$ for some $x'\sim x$ and $z'\sim z$, and since $\Di_a\psi \in \Gamma$,
$M,z'\mo\Di_a\psi$. Thus, $M,x'\mo \Di_a\Di_a\psi$. The formula $\Di_a\Di_a\psi\imp \Di_a\psi$ is true in $M$,  hence $M,x'\mo \Di_a\psi$.  So $M,x\mo \Di_a\psi$, as required. 
\end{proof}
\hide{
\IS{Stronger form:  $\ff{R}_a$ is contained in the Lemmon filtration, then
$(\ff{R}_a)^+$ is contained in
the maximal filtered relation $\maxR{a}$.
(Almost trivial: $(\ff{R}_a)^+$ is contained in Lemmon, since the latter is transitive)}
}

\subsection{Metric closure operations on  frames} 

The formula 
$\Di_c\Di  p \imp \Di_c p$ expresses the property $R_c\circ R \subseteq R_c$.
The following construction will be used to build filtrations with this property. 
Let $R,D$ be relations on $X$. We define recursively a relation $S$ called {\em the $R$-closure of $D$}  as follows: 
 \begin{eqnarray}
  S^{(0)} &=& D\cup D^{-1},\\
 S^{(n+1)}  &=&  (S^{(n)} \circ  R) \cup (S^{(n)}\circ R) ^{-1}, \\
   S  &=&  \bigcup\nolimits_{\omega} S^{(n)}.
 \end{eqnarray}
The relation  $S^{(n)}$ is called the {\em $n$-th grade} of the closure $S$.

\begin{lemma}\label{lem:poly:bcclosure}
$S$ is symmetric and $S\circ R \subseteq S$.
\end{lemma}
\begin{proof}
Each $S^{(n)}$ is symmetric, so is $S$.  Assume $(x,y)\in S\circ R $. Then
$(x,y)\in S^{(n)}\circ R $ for some $n$, and so $(x,y)\in S^{(n+1)}\subseteq S$.
\end{proof}


\hide{
HIDDEN COMMENT. IF IT APPEARED IN THE GALLEY PROOF, IT INDICATES THE PROBLEM.
For a real $r$, let
$V_r$ be the set of tuples $\vect{r}$ of parameters with $\sumvct{r}\leq r$,
and let $V=V_{\max \Par}$.\IS{This $V$ needs finiteness of $\AlA$.  Where do we use it? In the proof of filtration theorem. Move this definition.}
}

For a non-empty tuple $\vct=(r_1,\ldots,r_m)$ of parameters and 
a frame $(X, (S_r)_\Par)$, put
$S_\vect{r}=S_{r_1}\circ \ldots \circ S_{r_m}$. 
For the empty $\vect{r}$, define $S_\vect{r}$ as the diagonal $\{(x,x)\mid x\in X\}$.  
%

\begin{definition}
\label{defn: A-closure}
The {\em $\Par$-closure  of a frame $(X, (S_r)_\Par)$} is the frame
$(X, (H_r)_\Par)$ such that 
$$
H_{r}=\bigcup\{S_\vect{r}\mid \text{$\vect{r}$ is a tuple of parameters with $\sumvct{r}\leq r$}\}. 
$$
\end{definition}

\begin{lemma}\label{lem:poly:Pclosure}
The $\Par$-closure
$(X, (H_r)_\Par)$ of any frame $(X,(S_r)_\Par)$ validates the formulas
$\Di_{\vct} p\imp \Di_{r} p$  for each tuple $\vct$ over $\Par$ and $r\in \Par$ with $\sumv\leq r$.
\end{lemma}
\begin{proof}
Let $(x,y)\in H_\vct$ for $\vct=(r_1,\ldots,r_m)$.
Then we have $(x,y) \in S_{\vect{s}_1}\circ \ldots \circ S_{\vect{s}_m}$ for some $\vect{s}_i$ with $\sumvct{s}_i\leq r_i$.
Then $(x,y) \in S_\vect{s}$, where $\vect{s}$ is the concatenation $\vect{s}_1\vect{s}_2\ldots \vect{s}_m$.  Hence $\sumvct{s}\leq r$. By the definition of $H_r$, we have $(x,y)\in H_r$, as desired.
\end{proof}

\begin{definition}
The {\em $(\Par,\Di)$-closure  of a frame  $F=(X, R, (D_r)_\Par)$}  is the frame $G=(X,\preceq, (H_r)_\Par,X\times X)$ defined as follows:
\begin{enumerate}[(i)]
\item $\preceq$ is the transitive closure of $R$; 
\item $(X,(H_r)_\Par)$ is the $\Par$-closure of $(X,(S_r)_\Par)$, 
where $S_r$ is the $\preceq$-closure of $D_r$. 
\end{enumerate}
\end{definition}

\begin{lemma}\label{lem:poly-twoclosures}
For a frame $F=(X, R, (D_r)_\Par)$ with all relations reflexive,
its $(\Par,\Di)$-closure validates 
the logic $\LMetrAx{\Par,\Di }$. 
\end{lemma}

\begin{proof}
Let $G=(X,\preceq, (H_r)_\Par, X\times X)$ be this closure. The axioms for the modality $\EE$ hold trivially. Since $R$ is reflexive, $\preceq$ is a preorder.
%
That the $\Di_{\vct}\, p\imp \Di_{r} p$ are valid follows from Lemma \ref{lem:poly:Pclosure}.
All $S_r$ are symmetric by Lemma \ref{lem:poly:bcclosure}, hence
each $H_{r}$ is the union of compositions of symmetric relations, and so is symmetric as well.

Let us  check the validity of $\Di_r\Di p\imp \Di_r p$. For this, we  show that $H_r\, \circ \preceq \;\subseteq\; H_r$. In turn, it is enough to show that for a tuple of parameters $\vec{r}$ with $\oplus\,\vec{r}\leq r$ we have  $S_{\vec{r}}\,\circ\preceq\,\,\subseteq\, H_r$. 
For any $t\in A$ we have $S_t\circ{\preceq}$ included in $S_t$ by  Lemma \ref{lem:poly:bcclosure},
and $S_t$ included in  $S_t\circ {\preceq}$ due to reflexivity of $\preceq$.
So we have for any non-empty $\vect{r}=\vect{s}t$:
$$S_\vect{r}\circ {\preceq}\;=\;(S_\vect{s}\circ S_t)\circ {\preceq}\; = \; S_\vect{s}\circ (S_t \circ {\preceq})\;=\; S_\vect{s}\circ  S_t\;=\;  S_\vect{r}.$$
Thus when $\vect{r}$ is non-empty, $S_{\vec{r}}\;\circ\preceq\;\subseteq H_r$. When $\vec{r}$ is empty, $S_{\vec{r}}$ is by definition the identity relation, so we must show $\preceq\;\subseteq H_r$. Since $D_r$ is reflexive, so is $S_r$, hence $\preceq\;\subseteq S_r\,\circ\preceq$. Thus by Lemma~\ref{lem:poly:bcclosure} used again,  $\preceq\;\subseteq S_r$, and clearly $S_r\subseteq H_r$. 
\end{proof}

\subsection{The filtration}




The following is a corollary of 
\cite[Theorem 10]{Shehtman99} and \cite[Lemma 2]{ChromoWallic}.

\begin{lemma}
\label{lemma:poly:connFiltr}
Consider a model $M=(X,(R_i)_I,X\times X,\val)$ and let $\ff{X}= X{/}{\sim_\Delta}$ for a 
 finite 
$\Sub$-closed set of formulas $\Delta$. 
Let $a\in I$, and assume that a relation $\ff{R}$ on $\ff{X}$  includes the minimal 
filtered relation
$R_{a,\sim_\Delta}$. Assume $M\mo L$ for an $\Al$-logic $L$, and 
$L$ contains the formula $\EE p \wedge \EE \neg  p \imp \EE (\Di_a p\wedge \Di_a \neg p)$ or the formula 
$\EE p \wedge \EE \neg  p \imp \EE (p\wedge \Di_a \neg p)$.  
Then the frame $(\ff{X},\ff{R})$ 
is connected.
\end{lemma}  
\begin{proof}

Assume that $\ff{X}=\ff{Y}\cup\ff{Z}$ for disjoint non-empty $\ff{Y}$ and $\ff{Z}$, $Y=\bigcup \ff{Y}$, $Z=\bigcup \ff{Z}$. We aim to show that  
\begin{equation}\label{eq:conn-lemma}
\EE y\in Y\; \EE z\in Z \;([y]\, \ff{R}\, [z] \text{ or }[z]\, \ff{R}\, [y]).
\end{equation}

Since $\ff{Y}$ consists of $\sim_\Delta$-classes, for some Boolean combination $\gamma$ of 
formulas in $\Delta$ we have for all $y\in X$:
$$M,y\mo \gamma \quad \text{iff} \quad  y\in Y.$$ 

Assume that $L$ contains the formula  $\EE p \wedge \EE \neg  p \imp \EE (\Di_a p\wedge \Di_a \neg p)$. Then it also contains its substitution instance $\EE \gamma \wedge \EE \neg  \gamma \imp \EE (\Di_a \gamma\wedge \Di_a \neg \gamma)$. Since both $Y$ and $Z$ are non-empty, there are points $x_1$, $x_2$, and $x_3$ such that $x_1Rx_2$, $x_1Rx_3$, $M,x_2\mo\gamma$, and $M,x_3\mo\neg \gamma$. Hence, $x_2\in Y$ and $x_3\in Z$. If $x_1\in Y$, put $y=x_1$ and $z=x_3$; otherwise, put $y=x_2$ and $z=x_1$. In either case, we have \eqref{eq:conn-lemma} according to the definition of minimal filtration. 

Now assume that $L$ contains $\EE p \wedge \EE \neg  p \imp \EE (p\wedge \Di_a \neg p)$. Similar reasoning shows that there are points $y,z$ such that $yRz$, $M,y\mo\gamma$, and $M,z\mo\neg \gamma$, and so $y\in Y$ and $z\in Z$. 

Hence, we have \eqref{eq:conn-lemma}, and so $(\ff{X},\ff{R})$ is connected. 
\end{proof}

\begin{lemma}\label{thm:poly:fmp-top-manyD}
Assume $\AlA$ is finite. Let $L=\LMetrAx{\Par}+\ACon_a$ for some $a\in\Par$, or $L=\LMetrAx{\Par,\Di }+\ACon_a$ for some $a\in\Par\cup \{\Di\}$. Then $L$ has the finite model property.
\end{lemma}

\begin{proof}
The proof for the case when $\Di$ is in the language is more general, so we consider this case. 


In the case $\Alp=\emp$, $L$ is the logic $\LS{4UC}$, whose fmp is known \cite{Shehtman99}.  Assume  $\Alp\neq \emp$.

The proof is technical, and we first outline the strategy. Assume that $\vf$ is $L$-consistent. To establish the result, we find a finite model of $L$ in which $\vf$ is satisfiable. (A) Begin with a certain model $M$ of $L$ in which $\vf$ is satisfiable. (B) Take $\Gamma=\Sub\vf$ and $\Delta$ to be a certain finite set of formulas constructed from $\Gamma$. Then form the minimal filtration of $M$ over $\sim_\Delta$. 
(C) Take the $(A,\Di)$-closure of this filtration $\ff{F}$. By Lemma~\ref{lem:poly-twoclosures}, $\ff{F}$ validates $\LMetrAx{\Par,\Di }$. The details of our setup allow us to use  Lemma~\ref{lemma:poly:connFiltr} to obtain that $\ff{F}$ validates $\ACon_a$. Since $\Delta$ is finite, the minimal filtration is finite, hence $\ff{F}$ is finite. 
(D) Define a valuation to construct a model $\ff{M}$ from $\ff{F}$. Then show $\ff{M}$ is a filtration of $M$ through~$\Gamma$, hence by the filtration lemma $\vf$ is satisfible in $\ff{F}$.  
\smallskip

\noindent {\bf Step (A)} Assume $\vf$ is $L$-consistent. Then $\vf\in x_0$ for some $x_0$ in the canonical model of $L$. Let $M=(X,R, (D_r)_{\Par},R_\EE,\val)$ be the submodel of the canonical model of $L$ generated by $x_0$; here $R$ interprets $\Di$, $R_\EE$ interprets the modality $\EE$.
The logic $\LMetrAx{\Par,\Di }$ is canonical since its axioms are Sahlquist formulas. Hence, we have: $R$  is a preorder, all $D_r$  are symmetric and reflexive, and for all $r\in \Par$
\begin{equation}\label{eq:poly:closure-canonGen}
D_r\circ R\subseteq D_r.
\end{equation}
Since $M$ is point-generated,  we have $R_\EE=X\times X$.
\smallskip

\noindent {\bf Step (B)}
Let $\Gamma=\Sub\vf$ and let $V$ be the set of tuples $\vec{r}$ of parameters with $\oplus\vec{r}\leq r$ for some $r\in A$. Define 
\begin{eqnarray*}
\Delta &=&  \{\Di_\vect{r}\,\psi, \Di\Di_\vect{r}\psi \mid \vect{r}\in V , \psi \in \Gamma\}.
\end{eqnarray*}
Note that the empty string $\vec{r}$ has $\oplus\vec{r}=0$, so belongs to $V$, 
and in this case $\Di_{\vec{r}}\psi$ is $\psi$ for any formula. 
Thus $\Gamma\;\subseteq\;\Delta$. 
Put $\ff{X}=X/{\sim}$ for the equivalence $\sim\;=\;\sim_\Delta$ induced by $\Delta$. Then $\sim$ refines $\sim_\Gamma$. Let $R_\sim$ be the minimal filtered relation induced by $R$ on $\ff{X}$, and let $E_r$ denote the   minimal filtered relation $D_{r,\sim}$. 
\smallskip

\noindent {\bf Step (C)}
Let $\ff{F}=(\ff{X},\preceq, (H_r)_\Par,\ff{X}\times \ff{X})$ be the $(\Par,\Di)$-closure of $(\ff{X},R_\sim, (E_r)_\Par)$. Since $R$ and the $D_r$ are reflexive, so are their minimal filtrations $R_\sim$ and the $E_r$,  and  so by Lemma \ref{lem:poly-twoclosures}, $\ff{F}$ validates $\LMetrAx{\Par,\Di }$. 
Since $L$ contains $\ACon_a$, $\ff{F}$ validates $\ACon_a$ by Lemma \ref{lemma:poly:connFiltr}. Hence, $\ff{F}$ validates $L$.

\smallskip

\noindent {\bf Step (D)} Define $\ff{M}=(\ff{F},\ff{v})$, where $\ff{\val}(p)=\{[x]\in \ff{X} \mid x \in\val(p)\}$ for variables $p\in \Gamma$.  
Our aim is to show the following:
\[
\mbox{$\ff{M}=(\ff{X},\preceq, (H_r)_\Par,\ff{X}\times \ff{X},\ff{\val})$ is a filtration of $M=(X,R, (D_r)_{\Par},R_\EE,\val)$ through $\Gamma$.}
\]

By definition, $\ff{X}=X/\sim$, and as noted in Step~(B), $\sim$ refines $\sim_\Gamma$. The valuation $\ff{\val}$ is defined to match the criterion set in Definition~\ref{defn: filtration}. It remains to show that each relation of $\ff{M}$ lies between the minimal and maximal filtered relations. Trivially, $\ff{X}\times \ff{X}$ is the minimal filtered relation of $R_\exists=X\times X$. By definition of the $(A,\Di)$-closure, $\preceq$ is the transitive closure of $R_\sim$. So clearly $\preceq$ contains the minimal filtered relation $R_\sim$. 
Since $L$ contains $\Di\Di p\to \Di p$ and $M\mo L$,  we have $M\models \Di\Di\psi\to\Di\psi$ for any formula $\psi$,
so by Lemma~\ref{lemma:poly:transFiltr}, 
\begin{equation*} 
\text{ $\preceq$ is contained in $R_\sim^{(\Delta)}$.} 
\end{equation*}
But $\Gamma\subseteq\Delta$ gives $R_\sim^\Delta\subseteq R_\sim^\Gamma$.

It remains to show that for each $r\in A$, the relation $H_r$ lies between the minimal and maximal $\sim$-filtrations of $D_r$. Recall that $D_{r,\sim}$ is written $E_r$ and $H_r$ is constructed by first taking the $\preceq$-closure $S_r$ of $E_r$ as given before Lemma~\ref{lem:poly:bcclosure}, then obtaining $H_r$ from the family $(S_r)_\Par$ as in Definition~\ref{defn: A-closure}. 
Note that $S_r$ contains $D_{r,\sim}$ and $S_r\subseteq H_r$ is immediate from the definition of $H_r$. Thus $H_r$ contains the minimal $\sim$-filtration of $D_r$. To show that it is contained in the maximal filtration $D_{r,\sim}^\Gamma$ we must show 
\begin{equation}\label{eq: John}
\text{If $\Di_r\psi\in \Gamma$, $[x]\, H_r\, [y]$, and $M,y\mo \psi$, then $M,x\mo \Di_r\,\psi$.}
\end{equation}
We establish this through a series of results. 

We first claim that for $[x]\preceq [y]$:
\begin{eqnarray}
\label{eq:poly:transFiltrGenJohn}
&&\text{ if $\Di\psi\in \Delta$ and  $M,y\mo   \Di\psi$},  \text{ then $M,x\mo\Di \psi$;}\\
\label{eq:poly:closure-upGenJohn}
&&\text{ if $\Di_r\psi\in \Delta$ and  $M,x\mo   \Di_r\psi$},  \text{ then $M,y\mo  \Di_r\psi$. }
\end{eqnarray}

We have 
\eqref{eq:poly:transFiltrGenJohn} as a corollary of the fact that $\preceq\;\subseteq\;R_\sim^\Delta$. 
Indeed, if $M,y\mo   \Di\psi$,  then $M,z\mo \psi$ for some $z$ with $yRz$. Now $[x] \preceq [z]$, so $[x]\,R_\sim^\Delta\,[z]$, giving $M,x\mo\Di \psi$.

Let us check  \eqref{eq:poly:closure-upGenJohn}. By induction on $n$,  we show that \eqref{eq:poly:closure-upGenJohn}  holds for $[x] (R_\sim)^n [y]$ for all $n\geq 0$. The basis $n=0$ is trivial, since  $x\sim y$ in this case, meaning that $M,x\models\psi$ iff $M,y\models\psi$ for every formula $\psi$ in $\Delta$ and we assumed $\Di_r\psi\in \Delta$. Let $n>0$. We have  $[x]\, R_{\sim}\, [z]\, (R_{\sim})^{n-1} [y]$ for some $z$, and $x'R\, z'$ for some $x'\sim x$ and $z'\sim z$. Then $M,x'\mo \Di_r\psi$, and we have $M,u\mo \psi$ for some $u$ with $x' D_r\, u$. Since $D_r$ is symmetric, we have $u\,D_r\, x'$, and so $u\, D_r\circ R\, z'$. So $u\, D_r\, z'$ by \eqref{eq:poly:closure-canonGen}. Hence, $z' D_r\, u$, and so $M,z'\mo \Di_r \psi$. We have $[z'] (R_{\sim})^{n-1} [y]$, and $M,y\mo \Di_r \psi$ now follows from the induction hypothesis. This completes the proof of  \eqref{eq:poly:closure-upGenJohn}.
    
Now by induction on $n$, we show that for the $n$-grade $S_r^{(n)}$ of $S_r$, we have 
\begin{eqnarray}\label{eq:poly:closureMaxGenJohn}
&&\text{ if $\Di_r \psi\in \Delta$, $[x]\, S_r^{(n)}\, [y]$, and $M,y\mo   \Di\psi$},  \text{ then $M,x\mo\Di_r\psi$.}
\end{eqnarray}

First, observe that $\Di_r \psi\in \Delta$ gives $\Di\psi\in \Delta$. Indeed, if $\Di_r\psi\in \Delta$, then $\Di_r\psi$ has the form $\Di_r\Di_\vect{s}\,\chi$ for some $\chi\in \Gamma$ and $\vect{s}\in V$. Then $\Di\Di_\vect{s}\,\chi\in \Delta$.

Let $[x]\, S_r^{(0)}\,[y]$. We have $S_r^{(0)}=E_r  \cup E_r^{-1}$. Since $E_r$ is symmetric,  $S^{(0)}=E_r=D_{r,\sim}$. 
The assumptions of \eqref{eq:poly:closureMaxGenJohn} give some $x'\sim x$, $y'\sim y$ with $x' D_r\, y'$, and since $\Di\psi\in \Delta$, we have $M,y'\mo \Di\psi$. Hence, $M,x'\mo \Di_r\Di \psi$, and so $M,x'\mo \Di_r \psi$, since $\Di_r\Di \psi\imp \Di_r\psi$ is true in $M$. Thus, $M,x\mo \Di_r \psi$.

Now let $[x]\, S_r^{(n+1)}\,[y]$ and $M,y\models\Di\psi$. Consider two cases. First, assume that $[x]\,  S_r^{(n)} \circ {\preceq}\, [y]$. Then we have that $[x]\,  S_r^{(n)} [z] \; {\preceq}\; [y]$ for some $[z]$. So by \eqref{eq:poly:transFiltrGenJohn} we have $M,z\mo \Di \psi$. By induction hypothesis, we have $M,x\mo \Di_r \psi$. Now assume that $[x]\, (S_r^{(n)} \circ {\preceq})^{-1}\, [y]$, that is $[y]\, S_r^{(n)} [z] \preceq [x]$ for some $z$. Since $S_r^{(n)}$ is symmetric, we have $[z]\, S_r^{(n)}\, [y]$, and so $M,z\mo \Di_r\psi$ by the induction hypothesis. Now \eqref{eq:poly:closure-upGenJohn} implies that $M,x\mo \Di_r\psi$. This completes the proof of \eqref{eq:poly:closureMaxGenJohn}.

Using  \eqref{eq:poly:closureMaxGenJohn}, now we show:
\begin{eqnarray}\label{eq:poly:closureMaxGen1John}
&&\text{ If $\Di_r \psi\in \Delta$, $[x]\, S_r\, [y]$, and $M,y\mo \psi$},  \text{ then $M,x\mo\Di_r\psi$}.
\end{eqnarray}
Since $M,y\mo \psi$ and $R$ is reflexive, $M,y\mo \Di \psi$. We have $[x]\, S_r^{(n)}\,  [y]$ for some $n$, and by \eqref{eq:poly:closureMaxGenJohn}, $M,x\mo \Di_r\psi$, as desired.

Now we claim that  for $\vect{r}\in  V $, we have:
\begin{equation}\label{eq:poly:filtrGenJohn}
\text{If $\psi\in \Gamma$, $[x]\, S_\vect{r}\, [y]$, and $M,y\mo \psi$, then $M,x\mo \Di_\vect{r}\,\psi$.}
\end{equation}
By induction on the length of $\vect{r}$. If $\vect{r}$ is empty, $S_\vect{r}$ is the diagonal, so $[x]=[y]$; also, $\Di_\vect{r}\,\psi$ is just $\psi$; now \eqref{eq:poly:filtrGenJohn} holds because $\psi\in \Gamma\subseteq\Delta$. For the inductive step, suppose that $\vect{r}$ is the concatenation $r\vect{s}$ for a parameter $r$ and a tuple of parameters $\vect{s}$. We have $[x]\,S_{r}\, [z]\, S_\vect{s}\,[y]$ for some $z$. Notice that $\vect{s}\in V$, so by the induction hypothesis, we have $M,z \mo \Di_\vect{s}\,\psi$.  We have $\Di_r\Di_\vect{s}\,\psi\in \Delta$, and by \eqref{eq:poly:closureMaxGen1John}, we get  $M,x\mo \Di_\vect{r}\, \psi$. This completes the proof of  \eqref{eq:poly:filtrGenJohn}.

To establish \eqref{eq: John}, assume that $\Di_r\psi \in \Gamma$, $[x]\, H_r\, [y]$, and $M,y\mo\psi$. We have $[x]\,  S_\vect{r}\,  [y]$ for some tuple of parameters $\vect{r}$ with $\sumvct{r}\leq r$. Also, since $\Gamma$ is Sub-closed, $\psi\in \Gamma$. By \eqref{eq:poly:filtrGenJohn}, $M,x\mo \Di_{\vct}\, \psi$. We have $\sumv\leq r$,  so the formula $\Di_{\vct} \, p\imp \Di_{r} p$ is in $L$. So
$\Di_{\vct}\, \psi\,\imp \Di_{r} \psi$ is true in $M$.  Thus, $M,x\mo \Di_r\psi$, and establishing \eqref{eq: John}.

We have established that $\ff{M}$ is a filtration of $M$ through $\Gamma$. It then follows by the filtration lemma that $\vf$ is satisfiable in $\ff{F}$, completing the proof.  
\end{proof}

\ISLater{Now, to treat 0-connectedness (or even $\Di$?),  let $\Par=\mQ_{>0}$. For a given $\vf$, take the sequence $\Par_n$ with $\min \Par_n \to 0$.
We will get a chain of quotients $\ff{M}_n$, where:
\begin{itemize}
    \item for each small $r$, almost all will be $r$-connected, and the connectedness   will take the first-order form: some fixed (not arbitrary) degree of the filtered $r$-relation is universal;  (for $\Di$, connectedness will not take a first-order form!)
    \item at every $\ff{M}_n$ we will have $\EE x \,\vf^*$  for the first-order translation $\vf^*$ of $\vf$ (same for all $n$);
\item
on  $\ff{M}_n$ we will have a refining sequence of metric relations, each inducing a metric.
\end{itemize}

A ``minor'' step is to merge them all to a model over a 0-connected metric space. Never had a chance to check the most natural candidate... This is definitely deserve an attempt through the weekend.

Some refs:
\begin{itemize}
    \item second answer here: \url{https://mathoverflow.net/questions/270486/on-ultraproducts-of-topological-spaces} and links within
    \item Ultraproduct in topology \url{https://www.mssc.mu.edu/~paul/Paper/ult.pdf}
    (roughly, consider the base as unary predicates, use Los; in some sense, it is quite voluntaristic construction) For connected: see Example 1.6... is it a bad news? See also table on p. 21 for connected; path-connected $T_1$ (does not apply?); see Thm A2.2 and the remark below. Interesting corollary 2.7: there is an ultrapower of any metric X with an ultrametric.
    \item  Ultralimit of metric spaces \url{https://en.wikipedia.org/wiki/Ultralimit}; I do not understand this, a better reference is needed.
    \item Topological ultracoproduct etc \url{https://arxiv.org/pdf/math/9704205}
\end{itemize}

}


\section{The logic of graph-connected spaces}\label{sec:a-conn}

\begin{theorem}\label{thm:r-conn-compl}
Let $\Al$ be a set of positive real numbers, $a\in \Al$. 
The $\Alp$-logic of the class of all $a$-connected metric spaces, and also of all finite $a$-connected metric spaces, is $\LMetrAx{\Alp}+\ACon_a$. 
\end{theorem}

The proof of this theorem is based on two ingredients. The first one is Kripke completeness of the logic, which follows from the finite model property established earlier. The second ingredient is the following lemma, which in fact follows from  \cite{Kutz-Sturm-Suzuki-Wolter-Zakharyaschev2003,Kutz2007}.
\ISLater{Re-check: Kutz-Sturm-Suzuki-Wolter-Zakharyaschev2003 (looks like Lemma 4.12); the Wolter2005 ref; check earlier refs}
\begin{lemma}[Corollary of \cite{Kutz-Sturm-Suzuki-Wolter-Zakharyaschev2003,Kutz2007}]\label{lem:space-from-frame}
Assume that $\Alp$ is finite. Let $F=(X, (R_l)_{l\in \Alp},X\times X)$ be a $\LMetrAx{\Alp}$-frame,
and assume that the frame $(X, (R_l)_{l\in \Alp})$ is point-generated.
Then there exists a metric $d$ on $X$ such that $F$ is the $\Alp$-metric frame of $(X,d)$, and
$\inf\{d(a,b)\mid a\neq b\mbox{ and }a,b\in X\}>0$.
\end{lemma} 
This statement is only a slight modification of known facts; for the sake of rigor, we provide its proof in the Appendix. 

\ISLater{In the present text we do not need this: $\inf\{d(a,b)\mid a\neq b\,\&\,a,b\in X\}>0$. But let us keep it, just in case - might be useful later.}

\begin{proof}[Proof of theorem \ref{thm:r-conn-compl}.]
Soundness is straightforward. 

To prove the other inclusion, assume that a formula $\vf$ is consistent with the logic $\LMetrAx{\Alp}+\ACon_a$. Let $\AlB$ be the set consisting of $a$ and the parameters occurring in $\vf$, and let $L=\LMetrAx{\AlB}+\ACon_a$. 
Then $\vf$ is consistent with the logic $L$. 
Due to Lemma \ref{thm:poly:fmp-top-manyD}, $L$ has the finite model property and so is Kripke complete. 
Hence, $\vf$ is satisfiable in a (finite) Kripke frame $F=(X, (R_l)_{l\in \AlB},X\times X)$, 
which validates $L$. Due to $a$-connectedness, the frame $(X, (R_l)_{l\in \AlB})$ is point-generated. 
By Lemma \ref{lem:space-from-frame}, $F$ is the $\AlB$-metric frame of a metric space $(X,d)$. 
Clearly, this space is $a$-connected. 
Now consider the expansion $F_\Alp(X,d)$ of $F$.
The formula $\vf$ is satisfiable in $F_\Alp(X,d)$, which completes the proof.
\hide{

...

Due to the Kripke completeness of $\LMetrAx{\Alp}$ (Proposition \ref{prop:metr:canon}), $\vf$ is satisfiable in a $\LMetrAx{\Alp}$-frame $F=(X,(R_l)_{l\in\Alp})$.
Then $\vf$ is satisfiable in the reduct $G=(X,(R_l)_{l\in \AlpB})$ of $F$, where $\AlpB$ is the set of reals $l$ such that $\Di_l$ occurs
in $\vf$.
It follows that $\vf$ is satisfiable in a point-generated subframe $E$ of $G$.
Clearly, $G$ is a $\LMetrAx{\AlpB}$-frame. Hence, $E$ is a $\LMetrAx{\AlpB}$-frame too.
By Lemma \ref{lem:for-completeness}, $E$ is the $\AlpB$-metric frame of a metric space $(X,d)$. Consider the expansion $F_\Alp(X,d)$ of $E$.
The formula $\vf$ is satisfiable in $F_\Alp(X,d)$, which completes the proof.
}
\end{proof}

\begin{remark}
Theorem \ref{thm:r-conn-compl} generalizes the finite model property  of the logic  of $a$-connectedness (Lemma \ref{thm:poly:fmp-top-manyD}) 
for the case of infinite alphabets: indeed, finite metric spaces correspond to finite Kripke frames.   
\end{remark}
\begin{remark}
Theorem \ref{thm:r-conn-compl} does not apply (at least, directly) to the case of topological connectedness due to the following observation. If the closure modality is interpreted via a preorder in a Kripke frame, the corresponding topological space need not even be $T_1$. And hence, this topological closure is not induced by any metric.
This motivates our next section. 
\end{remark}

\ISLater{Conservativity?}

\section{The logic of connected spaces}\label{sec:top-conn}

\def\LCS{\LS{4UC_{<1}}}
\ISLater{
{\color{red} 
In this section, we consider the language with the topological modality, universal modality, and a single distance modality. Our aim is to axiomatize the logic of all connected metric spaces, and also of all connected and compact metric spaces, in this language. A primary tool is the work of Shehtman that uses the order-theoretic notion of a quasi-park to axiomatize the logic of connected topological spaces. We adapt these techniques to 

}
}

In this section, we consider the language with the topological modality, universal modality, and a single distance modality. Our aim is to axiomatize the logic of all connected metric spaces
in this language.

In the next subsection, we list the axioms of the logic, which we denote $\LCS$. 
We then discuss  
three preliminary constructions -- relational and topological (introduced earlier in \cite{Shehtman99}), 
as well as metric; they are given in subsections \ref{sub:quasi}, \ref{sub:cp}, and \ref{sub:jumps}.  
Then, for a given $\LCS$-consistent formula $\vf$, we construct a connected space where $\vf$ is satisfiable.  
This part of the proof is given in the two final subsections: 
we first construct a certain connected subspace $X$ of $\mR^3$; 
then, preserving its connectedness,  we alter its metric to 
ensure   satisfiability of $\vf$.

\subsection{The logic}

We assume that the alphabet $\Al$ of distance modalities is a singleton. 
Without loss of generality, we assume that $\Al=\{1\}$. 
Let $$\LCS=\LMetrAx{\{1\},\Di}+\ACon_\Di.$$
According to Definition \ref{def:metr-axioms}, the explicit set of axioms defining $\LCS$ is the following:
{
\setlength{\jot}{10pt}
\begin{eqnarray}
&&p\imp \AA \EE p,~p\imp\EE p,~\EE \EE p \imp \EE p, \text{ and }  \Di_1 p\imp \EE p;\\
\label{eq:S4UC1:sym}
&&p\imp \Box_1 \Di_1 p;\\
&&p\imp \Di_{1} p;\\
&&p\imp \Di p,~\Di \Di p\imp \Di  p; \\ 
\label{eq:S4UC1:cl}
&&\Di_1\Di p \imp \Di_1 p;\\
&&\EE p \wedge \EE \neg  p \imp \EE (\Di  p\wedge \Di  \neg p),  \text{ that is } \ACon_\Di.
\end{eqnarray}  
}

As we mentioned earlier, the counterpart of $\LCS$ for the empty $\Al$
is the logic $\LS{4UC}$.

\subsection{Quasiparks and suitable frames}\label{sub:quasi}

As usual, a {\em preorder} $(W,\preceq)$ is a set with a transitive reflexive relation. 
We use the following notation:
for $a\in W$ set $\ups a=\{b\mid a\preceq b\}$ and for $U\subseteq W$, set $\ups U =\bigcup_{a\in U}\ups a$. 
Likewise for downsets $\ds a$ and $\ds U$.
A rooted poset $(P,\leq)$ is a {\em (transitive) tree}, if $\ds a$ is a finite chain for all $a$ in $P$. 
A preorder $(W,\preceq)$ is a {\em quasitree}, if its skeleton (the associated partial order) is a tree. 

\begin{proposition}\label{prop:prop-of-suit}
Let $(W,\preceq, S, W\times W)$ be an $\LCS$-frame.  We have: 
\begin{enumerate}
\item $\preceq$ is included in $S$;
\item For $a\in W$, $\ups a$ is an $S$-clique, that is $\ups a\times \ups a\subseteq S$;
\item \label{item3:prop:prop-of-suit} If
$a \preceq a'$, $b \preceq b'$, and  $aSb$, then 
$a'S b'$.   
\end{enumerate}
\end{proposition}
\begin{proof} Axioms~
\eqref{eq:S4UC1:sym} -- \eqref{eq:S4UC1:cl} give that $\preceq$ is a quasi-order, $S$ is reflexive and symmetric, and $S\,\circ\preceq\;\subseteq S$. So clearly we have (1)~ $\preceq\;\subseteq S$. For~(2) suppose $a\preceq b,c$. Then $b\,S\,a\,\preceq\,c$, so $(b,c)\in S\,\circ\preceq\;\subseteq S$. For (3)~assume $a\;\preceq\;a'$, $b\;\preceq\;b'$ and $a\,S\,b$. Then $b\,S\,a\,\preceq\,a'$, so $b\,S\,a'$, and then $a'\,S\,b$. Then $a'\,S\,b\,\preceq\,b'$ gives $a'\,S\,b'$. 
\end{proof}

\def\imm{ \lessdot}
\begin{definition}\cite{Shehtman99}
In a poset $(P,\leq)$, let 
$\imm$ denote the immediate successor relation: $x\,\imm\, y$ iff $x<y$ but $x <z< y$ for no $z$. A {\em non-trivial loop} in $(P,\leq)$ is a non-oriented $\imm\,$-path $(x_1,\ldots ,x_n,x_1)$
(in other words, with $x_i\,\imm\; x_{i+1}$ or $x_{i+1}\,\imm\; x_i$ for each $i\leq n$)
with $n\geq 3$ and all $x_1,\ldots x_n$ distinct.

A finite preorder $(W,\preceq)$ is called a {\em quasipark}, if:
\begin{enumerate}[(i)]
\item 
the skeleton of $(W,\preceq)$ has no non-trivial loops, and 
\item one can order the minimal clusters of  $(W,\preceq)$  as $C_1,\ldots,C_n$ so that 
\begin{equation}\label{eq:quasipark}
    \ups C_i\cap \ups C_{j}\neq\emp \text{ iff }|i-j|\leq 1.
\end{equation}
\end{enumerate}
\end{definition}


The following is straightforward from the definition.    
\begin{proposition}\label{prop:quasipark-basic}
In a  quasipark  $Q=(W,\preceq)$, for each $a,b\in W$ we have:
\begin{enumerate}
    \item The restriction of $Q$ to $\ups a$ is a quasitree; 
    \item If $\ups a\cap \ups b$ is non-empty, then the restriction of $Q$ to this set is a quasitree.
\end{enumerate}
\end{proposition}

\begin{definition}
An $\LCS$-frame $(W,\preceq, S, W\times W)$ is said to be {\em suitable}, 
if $(W,\preceq)$  is a quasipark.
\end{definition}

It is known that $\LS{4UC}$ is characterized by quasiparks (endowed with the universal relation) 
\cite[Theorem 14]{Shehtman99}.
Together with the   finite model property of $\LCS$, this results in the following fact.
\begin{corollary}\label{cor:suitable}
  $\LCS$ is the logic of suitable frames.  
\end{corollary}
\begin{proof}
According to the finite model property of $\LCS$ (Lemma \ref{thm:poly:fmp-top-manyD}), 
it is enough  to show that any finite $\LCS$-frame $F=(X,\sqsubseteq, T, X\times X)$ is a p-morphic image of a  suitable $G$.

In view of \cite[Lemma 13]{Shehtman99}, there is a quasipark $(W,\preceq)$ and a p-morphism $f: (W,\preceq) \toto (X,\sqsubseteq)$. 
Set $S=\{(a,b)\in W\times W\mid f(a) T f(b)\}$, $G=(W,\preceq, S, W\times W)$. 
It is immediate that $f:G\toto F$. Also, it is  straightforward that $G$ is an $\LCS$-frame. In particular, the formula $\Di_1\Di p \imp \Di_1 p$ holds there. For this, assume 
that we have $a S b \preceq c$. Then $f(a)Tf(b)\sqsubseteq f(c)$, and hence $f(a) T f(c)$.
By the definition of $S$, $aS c$. 
\end{proof}

\subsection{Morphisms from topological spaces to preorders
%
%
%
}\label{sub:cp}


In a topological space $(X,\tau)$, for $Y\subseteq X$, let $\cl{Y}$ denote the closure of $Y$.

\begin{definition}\cite{Shehtman99}
Let $(X,\tau)$ be a topological space, $(W,\preceq)$ a preorder. 
A surjective $f:X\to W$ is called a {\em cp-morphism},
if for each $a\in W$ we have 
$$
\cl f^{-1}(a)\;=\;f^{-1}[\ds a]
$$
(equivalently, $f$ is an interior map 
from $(X,\tau)$ onto $(W,\tau_\preceq)$, where $\tau_\preceq$ is the topology of upsets in $(W,\preceq)$.)
\end{definition}


The next statement is a corollary of \cite{McKinsey-Tarski-1944}; in  explicit form, it is given in  \cite[Lemma 16]{Shehtman99}.
\begin{lemma}\label{lem:ShehtmanLem16}
If $X$ is a
    connected 
dense-in-itself 
separable metric space and $F$ is a finite quasitree, then $f:X\toto F$ for a cp-morphism $f$.
\end{lemma}

\subsection{Jumps}\label{sub:jumps}

\def\j{\mathtt{j}}
\def\j{\iota}
\def\j{\mathtt{J}}
\def\j{J}

Here we introduce a notion of jumps -- a geometric tool to produce from a metric space $(X,d)$ a new metric space $(X,d_\j)$
that is topologically equivalent, 
but reduces distances  by making ``wormholes'' in space.

\begin{definition}\label{def:jumps}
Let $(X,\d)$ be a metric space, $E\subseteq X\times X$, and $\j:E\to\mR_{>0}$ a map that is symmetric in that if $(x,y)\in E$, then $(y,x)\in E$ and $J(x,y)=J(y,x)$.
Call $\j(x,y)$ the  {\em jump between } $x$ and $y$. 
For $x,y\in X$, define the {\em weight of $(x,y)$}: if $(x,y)\in E$ and $\j(x,y)<\d(x,y)$, put $w(x,y)=\j(x,y)$; otherwise, put $w(x,y)=\d(x,y)$. A {\em path in $X$ from $x$ to $y$} is a finite non-empty tuple $\tau=(x_0,\ldots,x_m)$ of elements of $X$ with $x_0=x$ and $x_m=y$.
Put $w(\tau)=\sum_{i<m} w(x_i,x_{i+1})$.
For $x,y \in X$, define
\begin{equation}\label{eq:distance-for-jumps}
\d_{\j}(x,y)=\inf\{w(\tau) \mid \tau\text{ is a path from $x$ to $y$}\}.
\end{equation}
\end{definition}

\begin{lemma}\label{lem:JumpsTop}
If   $\inf\{\j(x,y)\mid (x,y)\in E\}>0$, then $d_\j$ is a metric on $X$ and the 
topologies on $X$ induced by $\d$ and $\d_{\j}$ are equal.
\end{lemma}
\begin{proof}
It is easy to see that $\d_{\j}$ is a metric on $X$; 
in particular, the triangle inequality is straightforward from \eqref{eq:distance-for-jumps}. Also, open balls of radii less than this infimum 
form bases of both these topologies.
\end{proof}

\subsection{Constructing a connected space for $\vf$}\label{sub:spaceX}

Assume that a formula $\vf$ is $\LCS$-satisfiable. By Corollary \ref{cor:suitable}, $\vf$ is satisfiable in a suitable frame $$F=(W,\preceq, S, W\times W).$$ Our goal is to show that $\vf$ is satisfiable in  a connected metric space. 

First, we name parts of $F$.  
Let $C_i$, $1\leq i\leq n$ be the minimal clusters of $W$ in the order satisfying  \eqref{eq:quasipark}. We denote $\ups C_i$ as $W_i$, and the restriction of $(W,\preceq)$ to $W_i$ as $F_i$. Put $V_i=W_i\cap W_{i+1}$ (here $1\leq i<n$). 
By 
Proposition \ref{prop:quasipark-basic},
each $V_i$ has a least cluster  $D_i$. 
Let $d_i$ be some point in $D_i$. Finally, let $Q_i$ be the restriction of  $(W,\preceq)$ to $\ups D_i$.

\begin{figure}[h]
\begin{center}
    \label{fig:frameF}
\begin{tikzpicture}[
    cl/.style={circle, draw, thick, inner sep=0pt, minimum size=15pt},
    every edge/.style={thick},
    scale=0.6,
]

\node[cl, label={[font=\normalsize]below:$C_1$}] (C1) at (-4.5, 0) {};
\node[cl, label={[font=\normalsize]below:$C_2$}] (C2) at (0, 0) {};
\node[cl, label={[font=\normalsize]below:$C_3$}] (C3) at (4.5, 0) {};

\node[cl] (L1) at (-6.5, 2.8) {};
\node[cl] (L2) at (-4.5, 2.8) {};
\node[cl, label={[font=\normalsize]left:$D_1$}] (D1) at (-1.8, 2.8) {};
\node[cl, label={[font=\normalsize]right:$D_2$}] (D2) at (1.8, 2.8) {};
\node[cl] (R1) at (4.5, 2.8) {};
\node[cl] (R2) at (6.5, 2.8) {};

\node[cl] (D1a) at (-2.8, 5.6) {};
\node[cl] (D1b) at (-0.8, 5.6) {};
\node[cl] (D2a) at (0.8, 5.6) {};
\node[cl] (D2b) at (2.8, 5.6) {};

\draw[thick] (C1) -- (L1);
\draw[thick] (C1) -- (L2);
\draw[thick] (C1) -- (D1);
\draw[thick] (C2) -- (D1);
\draw[thick] (C2) -- (D2);
\draw[thick] (C3) -- (D2);
\draw[thick] (C3) -- (R1);
\draw[thick] (C3) -- (R2);

\draw[thick] (D1) -- (D1a);
\draw[thick] (D1) -- (D1b);
\draw[thick] (D2) -- (D2a);
\draw[thick] (D2) -- (D2b);

\end{tikzpicture}      

\vskip 10pt

\begin{tikzpicture}[
    disk/.style={draw, thick, fill=gray!10},
    scale=0.55,
]

\def\xA{-5}   
\def\xB{0}    
\def\xC{5}    
\def\ry{1.8}  
\def\rx{0.5}  

\draw[disk] (\xA,0) ellipse ({\rx} and {\ry});
\draw[disk] (\xB,0) ellipse ({\rx} and {\ry});
\draw[disk] (\xC,0) ellipse ({\rx} and {\ry});

\pgfmathsetmacro{\yAx}{\xA + 0.6*\rx*cos(30)}
\pgfmathsetmacro{\yAy}{0.6*\ry*sin(30)}
\fill (\yAx, \yAy) circle (2pt);
\node[above right, font=\small] at (\yAx, \yAy) {$y_1$};

\pgfmathsetmacro{\zAx}{\xB - 0.6*\rx*cos(20)}
\pgfmathsetmacro{\zAy}{-0.6*\ry*sin(20)}
\fill (\zAx, \zAy) circle (2pt);
\node[below left, font=\small] at (\zAx, \zAy) {$z_1$};

\pgfmathsetmacro{\yBx}{\xB + 0.6*\rx*cos(25)}
\pgfmathsetmacro{\yBy}{0.6*\ry*sin(25)}
\fill (\yBx, \yBy) circle (2pt);
\node[above right, font=\small] at (\yBx, \yBy) {$y_2$};

\pgfmathsetmacro{\zBx}{\xC - 0.6*\rx*cos(15)}
\pgfmathsetmacro{\zBy}{-0.6*\ry*sin(15)}
\fill (\zBx, \zBy) circle (2pt);
\node[below left, font=\small] at (\zBx, \zBy) {$z_2$};

\draw[thick] (\yAx, \yAy) -- (\zAx, \zAy);
\draw[thick] (\yBx, \yBy) -- (\zBx, \zBy);

\node[above, font=\small] at ({(\yAx+\zAx)/2}, {(\yAy+\zAy)/2 + 0.15}) {$H_1$};
\node[above, font=\small] at ({(\yBx+\zBx)/2}, {(\yBy+\zBy)/2 + 0.15}) {$H_2$};

\node[below, font=\normalsize] at (\xA, -\ry - 0.3) {$X_1$};
\node[below, font=\normalsize] at (\xB, -\ry - 0.3) {$X_2$};
\node[below, font=\normalsize] at (\xC, -\ry - 0.3) {$X_3$};

\end{tikzpicture} 

    \caption{A quasipark $(W,\preceq)$ and its preimage $X$}
\end{center}
\end{figure}
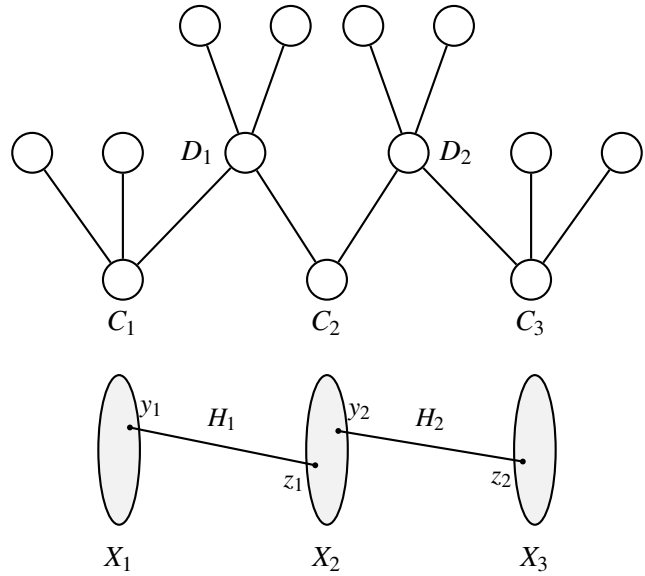


\smallskip 
Now we construct a subspace $X$ of $\mR^3$.


Fix a number $L>2$. 
Let $1\leq i\leq n$. In the plane $x= Li$, consider a closed disk $X_i$ centered at $(Li,0,0)$ of radius $\frac{1}{L}$,
and let $\tau_i$ be the standard topology on $X_i$.
Fix a cp-morphism $f_i:(X_i,\tau_i)\toto F_i$;
the existence of $f_i$ follows from  Lemma \ref{lem:ShehtmanLem16}.

Choose  $y_i$ in the $f_i$-preimage of $d_i$,
and $z_i$ in the $f_{i+1}$-preimage of $d_i$.
Define the {\em $i$-th handle} $H_i$ as the open interval between
between $y_i$ and $z_{i}$ on the line through these two points;  the standard topology on $H_i$ is denoted as $\rho_i$.
Fix a cp-morphism $g_i:(H_i,\rho_i)\toto Q_i$ (we use \cite[Lemma 16]{Shehtman99} again). 

Let $X$ be the (disjoint, in fact) union of these disks and handles, $\tau$ the standard topology on $X$.
Observe that $(X,\tau)$ is connected.  
Now let $f \; =\; f_1\cup g_1\cup f_2\cup \ldots \cup g_{n-1}\cup f_n$.

\begin{proposition}\label{prop:cpmorphBasic}
$f:  (X,\tau)\toto (W,\preceq)$  is a cp-morphism.
\end{proposition}  
\begin{proof}
Clearly $f$ is onto. Note that any set closed in $X_i$ is closed in $X$ and any set closed in $H_i$ and containing its endpoints is closed in $X$. So for any $a\in W$ we have that $f^{-1}[\ds a]$ is a finite union of closed sets of $X$, so is closed in $X$. Thus $\cl f^{-1}(a)\subseteq f^{-1}[\ds a]$. Conversely, suppose $x\in f^{-1}[\ds a]$ and consider two cases. Case (i) $x\in X_i$. Then $f(x)\in W_i$, and since $W_i$ is an upset, $a\in W_i$. Then $x$ is in the closure of $f_i^{-1}(a)$ in $X_i$, hence it is in the closure of $f^{-1}(a)$ in $X$. Case (ii) $x\in H_i$. Then $a\in Q_i$, and $x$ is in the closure of $g_i^{-1}(a)$ in $H_i$, hence in the closure of $f^{-1}(a)$ in $X$. 
\end{proof}



We say that two points $x,y$ in $X$ are {\em neighbours}, if $\d(x,y)<L$.
\begin{proposition}\label{prop:neighbours}
Let $x,y$ be neighbours. Then:
\begin{enumerate}
\item $\ds f(x)\,\cap \,\ds f(y)\neq \emp$;
\item $f(x)\,S\,  f(y)$;
\item $\cl{f^{-1}(f(x))}\,\cap\, \cl{f^{-1}(f(y))}\neq \emp$.
\end{enumerate}
\end{proposition}
\begin{proof}
Due to the construction, a pair of adjacent handles, as well as a pair of adjacent handle and a disk, is mapped to a quasitree. Hence,
if $\d(x,y)<L$, then $f(x)$ and $f(y)$ belong to the same quasitree $F_i$. 

That $f(x)S  f(y)$ follows from the fact that $F_i$ is an $S$-clique by Proposition \ref{prop:prop-of-suit}.

The third statement follows from the first, since $f$ is a cp-morphism.
\end{proof}
\ISLater{Nov. 13 2024: it seems that there is an extra axiom that is valid in this structure. It is unclear if it pertains to compactness or some other property. Making disks open does not seem to help.}

\ISLater{March 3, 2026: Our disks-handles model is 2-locally component connected (pinning out a point in a neighborhood gives at most two connected components)! This is known to be modally expressible in the language of the Cantor's derivative for $\Di$; but of course it indicates the source of our problems!  
 }

\subsection{Altering the metric} \label{sub:spaceJ}
Let $\d$ denote the Euclidean distance in $\mR^3$.  
Our goal is to define jumps $\j$ in $X$ in such a way that   
$\vf$ is satisfiable in $(X,\dj)$.

\smallskip 

Let $B(x,\delta)$   denote the open 
ball
in $X$ centered at $x$ of radius $\delta$ (w.r.t. the Euclidean distance).

\def\K{\frac{1}{L}}
Let $x\in X$. We put
\begin{equation}\label{eq:delta-x}
    \delta(x)=\sup\{\delta<\K  \mid
B(x,\delta)\subseteq f^{-1}[\ups f(x)]\}.
\end{equation} 
Since $f^{-1}[\ups f(x)]$ is open, $\delta(x)>0$. 

For $b\in W$, we define $y_b$ and $\delta(b)$. 
Let $i$ be the least such that $b \in W_i$.  Take $y_b$ in $f_i^{-1}(b)$, which is not a limit of a handle (notice that if $y\in f_i^{-1}(b)$ is a limit of a handle, then $b=d_i$ or $b=d_{i-1}$; in both cases, $\ds b\cap W_i$ contains the minimal points in $W_i$, which are distinct from $b$; we have
$\cl f_i^{-1}(b)=f_i^{-1}[\ds b\cap W_i]\neq f_i^{-1}(b)$, so $f_i^{-1}(b)$ is infinite.)

Let $\delta(b)$ be the supremum of $\delta$ such that $\delta<\delta(y_b)$ and $B(y_b,\delta)$
does not contain limits of handles.
Let $\delta_F= \frac{1}{2}\min\{\delta(b)\mid b\in W\}$.  So we have
   \begin{equation}\label{eq:delta-y}
     f[B(y_b,\delta_F)]\subseteq \ups b
   \end{equation}
   Also, notice that 
   $B(y_b,\delta_F)$ and $B(y_c,\delta_F)$ are disjoint for distinct $b,c\in W$.  

\ISLater{\IS{How do we use it?} {\color{red}[Maybe we don't]}}

Assume that for $x\in X$ we have $f(x)\, S\, b$.
If $\d(x,y_b)\geq 1$,  then we add the jump
$$\j(x,y_b)=1-\frac{1}{2}{\min\{\delta(x),\delta_F\}}.$$
Notice that $\j(x,y_b)>\frac{1}{2}$. 
\ISLater{This $\frac{1}{2}$ seems to be overkill...   }

\ISLater{
Old comment (for two modalities):  

To take fixed witnesses $y_b$ is provably WRONG (p.7)
}

Now consider the metric space $(X,\dj)$. 
\begin{proposition}
$(X,\dj)$ is connected. 
\end{proposition}
\begin{proof}
Follows from Lemma \ref{lem:JumpsTop} immediately.  
\end{proof}

\begin{proposition}
$f:(X,\dj)\toto  (W,\preceq)$ is a cp-morphism.
\end{proposition}
\begin{proof}
Follows from Lemma \ref{lem:JumpsTop} and Proposition \ref{prop:cpmorphBasic} immediately.  
\end{proof}

Now consider the distance relation on $(X,\dj)$: 
\[
x\,S^{(\j)}\, y\,\,\Leftrightarrow\,\, \dj(x,y)<1.
\]  

\begin{proposition} \label{prop:jump-p-morph}
$f:(X,S^{(\j)})\toto  (W,S)$ is a p-morphism.
\end{proposition} 
\begin{proof}
That $f$ satisfies the back property of $p$-morphism 
is immediate from the definition of $\j$: if $f(x)\,S\,b$ 
and $\d(x,y_b)\geq 1$,
we added the jump $\j(x,y_b)$.

To check the relational homomorphism property, assume 
$\dj(x,y)<1$. We aim to show that $f(x)Sf(y)$.

Let $w$ be the weighting function induced by $\j$ on $(X,\d)$ (Definition \ref{def:jumps}).
Then for some path $\tau=(x_0,\ldots,x_m)$ from $x$ to $y$ we have 
$\sum_{i<m} w(x_i,x_{i+1})<1$.

Clearly, this path has at most one jump, since any jump value is greater than $\frac{1}{2}$. 

If there are no jumps, then $d(x,y)<1$, and hence $x$ 
and $y$ are neighbours, and the claim follows from Proposition \ref{prop:neighbours}.

Assume that $\tau$ has one jump $(x_k,x_{k+1})$.  Then 
$x_{k+1}$ (or $x_{k}$) is $y_b$ for some $b$. It follows that 
both $\sum_{i<k} w(x_i,x_{i+1})$ 
and
$ \sum_{k+1\leq i<m} w(x_i,x_{i+1})$  are less than ${\min\{\delta(x_k),\delta_F\}}$.  
At the same time, these initial and final parts of $\tau$ have no jumps, and we have  $d(x,x_k)\leq \sum_{i<k} w(x_i,x_{i+1})$ and 
$d(x_{k+1},y)\leq \sum_{k+1\leq i<m} w(x_i,x_{i+1})$. 
\ISLater{$1/2$? Ok, $\leq 1/2$, so ``less then''. But still, $\sum_{i<m} w(x_i,x_{i+1})<1$, not $\leq$. Very tempting to remove $\frac{1}{2}$ from jump def. }

Now it follows from \eqref{eq:delta-x} and \eqref{eq:delta-y}  that 
$f(x_k) \preceq f(x)$ and $f(x_{k+1}) \preceq f(y)$. 
We also have $f(x_k)Sf(x_{k+1})$ due to the construction of jumps. 
By Proposition \ref{prop:prop-of-suit}\eqref{item3:prop:prop-of-suit}, it follows that 
$f(x)S f(y)$.   
\end{proof}

From the p-morphism lemma \cite[Section 3.3]{BDV} and cp-morphism lemma \cite[Lemma 15]{Shehtman99}, it follows that $\vf$ is satisfiable in $(X,d_\j)$.  
Hence, we have
\begin{theorem}
In the language with the topological modality, universal modality, and a single distance modality,    
the logic of all connected, and also of all  connected and compact metric spaces is  $\LCS$. 
\end{theorem}

\section{Concluding remarks} 

The construction given in the previous section does not seem to be applicable to the case of more than one metric modality. 
The main (and the only) problem is the relational homomorphism condition in the proof of Proposition 
\ref{prop:jump-p-morph}: our argument that any path was augmented with at most one jump 
does not extend to the case of several distance relations.  
Hence, the  problem of axiomatization of connected metric spaces  (\cite{Wolter2005,KuruczWZ05}) 
in the language of the topological closure $\Di$, universal modality, and multiple metric modalities remains open.

Another natural question is the axiomatization 
of metric spaces that are $a$-connected for all positive $a$. 
Such spaces are said to be {\em well-chained}  \cite{TopologicalAnalysis-Whyburn}.
We conjecture that for any set $\Al$ of parameters, the $\Alp$-logic of the class of all well-chained metric spaces is 
$\LMetrAx{\Alp}+\{\ACon_a\mid a\in \Alp\}$. 

\medskip
These are only two of many open problems related to modal axiomatization of metric spaces. 
Of particular interest are distance logics of Euclidean spaces.
Even the unimodal language with a single distance modality turns out to be surprisingly expressive: for example, it distinguishes between different dimensions and between the rational and real lines \cite{RTG-DistanceLogics}. 
However, no complete axiomatizations are known for the corresponding  logics. 

\section*{Acknowledgements}  
We are grateful to anonymous reviewers for their helpful comments on an earlier version of this text.

\smallskip \noindent
This work was supported by NSF Grant DMS - 2231414.

\section*{Appendix}

\begin{proof}[Proof of Lemma \ref{lem:space-from-frame}] 
Very close to the argument in the proof of \cite[Theorem 4.4]{Kutz2007}.

The case $\Alp=\emp$ is trivial. Assume  $\Alp\neq \emp$.
Let $R=\bigcup_{l\in \Alp} R_l$. Since $(X, (R_l)_{l\in \Alp})$ is point-generated, $(X,R)$ is connected.

Let $$k=\frac{\max{\Alp}+1}{\min{\Alp}},$$
and let  $\Alp^{<2k}$ 
be the set of tuples over $\Alp$ whose length is less than $2k$. 
Consider the set $$S=\{\sumv-l\mid \text{$\vct\in A^{<2k}$, $l\in \Alp$, and $\sumv> l$}\}.$$
Clearly, $S$ is a finite set of positive numbers. Let
\begin{equation*}\label{eq:epsilon}
\eps=\frac{\min (S\cup \Alp)}{2k}.
\end{equation*}
Since $k>1$ and $\min (S\cup \Alp)\leq \min A$, we have
\begin{equation}\label{eq:metr:eps-small}
\eps<\frac{\min \Alp}{2}.
\end{equation}

\hide{
We have:
\begin{equation}\label{eq:metr:m-small}
\text{If $\vct\in A^m$, $r\in \Alp$ and $\sumv-m\eps\leq r$, then $m<2k$.}
\end{equation}
Indeed, in view of \eqref{eq:metr:eps-small}, we have
$$
m\frac{\min A}{2}<m(\min A-\eps)\leq
\sumv-m\eps\leq
r\leq \max A,
$$
and hence we have
$$m<2\frac{\max{A}}{\min{A}}<2k.
$$
We also have:
\begin{equation}\label{eq:metr:main}
\text{If $m<2k$, $\vct\in A^m$, $r\in \Alp$ and $\sumv\geq r$, then $\sumv-m\eps\geq r$.}
\end{equation}
Indeed, $\sumv-r\in S$, and hence
$$
\sumv-r\geq \min S\geq  \min (S\cup A)=2k\eps>m\eps,
$$
so $\sumv-m\eps\geq r$.
}

Let $e=(a,b) \in R$, $a\neq b$. Define the {\em index $\idx(e)$ of $e$} as $\min\{l\in \Alp \mid e\in R_l\}$,
and its {\em weight} $w(e)$ as $\idx(e)-\eps.$
By the definition, we put $w(a,a)=0$.

For an $R$-path $\tau=(a_0,\ldots,a_m)$,
put
$w(\tau)=\sum_{i<m} w(a_i,a_{i+1})$.

Finally, for $a,b \in X$, we define
\begin{equation}\label{eq:distance-for-postgraph}
d(a,b)=\inf\{w(\tau) \mid \tau \text{ is an $R$-path from $a$ to $b$}\}.
\end{equation}

It is easy to see that $d$ is a metric on $X$. In particular, since $(X,R)$ is connected, $d(a,b)$ is defined for all $a,b\in X$.
That $d(a,b)=d(b,a)$ is immediate from the symmetry of relations.
The triangle inequality is straightforward from \eqref{eq:distance-for-postgraph}.
It is also clear from the definition of $w$ that $d(a,b)\geq \min \Alp-\eps$ for distinct $a,b$.

It remains to show that for $l\in \Alp$,
\begin{equation}\label{eq:metr-induced}
aR_l b \text{ iff } d(a,b)<l.
\end{equation}

The `only if' follows from \eqref{eq:distance-for-postgraph} and the definition of weight.

For the `if' direction, assume that $\displaystyle d(a,b)<l\in\Alp$. Then $a$ and $b$ are connected by an $R$-path $\tau=(a_0,\ldots,a_m)$ with
$w(\tau)<l$. The case $a=b$ is trivial, so let us assume that $m>0$.
We can also assume that $\tau$ is simple, that is all $a_i$ are distinct.
For $i<m$, put $l_i=\idx(a_i,a_{i+1})$, and consider $\vct=(l_0,\ldots, l_{m-1})$.
We have
\begin{equation}\label{eq:metr:weight-eps}
w(\tau)=\sum_{i<m}{(l_i-\eps)}=\sumv-m\eps<l.
\end{equation}

We claim that $m<2k$.
Indeed, in view of \eqref{eq:metr:eps-small}, we have
$$
m\frac{\min \Alp}{2}<m(\min \Alp-\eps)\leq
\sumv-m\eps=w(\tau)<
l\leq \max \Alp,
$$
and hence we have $m<2\frac{\max{\Alp}}{\min{\Alp}}<2k$.

Now we claim that
\begin{equation}\label{eq:metr:main}
\sumv\leq l.
\end{equation}
For the sake of contradiction, assume
$\sumv> l$. 
We have $\sumv-l\in S$, and hence
$$
\sumv-l\geq \min S\geq  \min (S\cup A)=2k\eps>m\eps,
$$
so $\sumv-m\eps> l$. This contradicts \eqref{eq:metr:weight-eps}, which  proves \eqref{eq:metr:main}.

We have $(a,b)\in R_{\vct}$. Since $F$ is a
$\LMetrAx{\Alp}$-frame and in view of \eqref{eq:metr:main}, $(a,b)\in R_{l}$. This completes the proof of \eqref{eq:metr-induced}.
\end{proof}

\ISLater{ $\Alp^{<2k}$  might be confusing, since $k$ is real} 

\ISLater{There must be a common term for $d$. ``Weighted shortest-path distance''?} 

\ISLater{Older staff:

\begin{corollary}[Conservativity]\label{lem:conservativity}
Let $\AlpB\subseteq \Alp$.
Then $\LMetrAx{\Alp}\cap \Fms_\AlpB=\LMetrAx{\AlpB}.$
\end{corollary}

\hide{
\begin{problem}
For which infinite $\Alp$, the class of all metric $\Alp$-frames is first-order axiomatizable?
\IS{Most probably, no such $\Alp$ can be unbounded: metric frames should be connected w.r.t. $\bigcup_\Alp{R_l}$.}
\end{problem}

}

\bigskip

\noindent
\ISLater{ Dangerous bend! The following version of the Lemma might be wrong:

Let $(X, (R_l)_{l\in \Alp})$ be a point-generated $\LMetrAx{\Alp}$-frame.
The for every finite $\AlpB\subseteq \Alp$, there exists a metric $d$ on $X$ such that
$R_l=d_l$ for each $l\in \AlpB$.

}
} 
\bibliographystyle{eptcs}
\bibliography{metric}

\hide{
\newpage

\begin{center}
{\huge \bf{Fragmented staff: unsorted lemmas, ideas, wrong statements... }}
\end{center}

\begin{center}
\bf{Before Nov. 15}
\end{center}

\input{sections/storeroom/_storeroom}

}
 
\end{document}